\documentclass[11pt,a4paper]{llncs}

\usepackage{fullpage}
\usepackage{graphicx}
\usepackage{amsmath}
\usepackage{amssymb}
\usepackage{algorithm}
\usepackage{algorithmic}
\usepackage{hyperref}

\newcommand{\Order}{\mathrm{O}}

\newcommand{\poly}{\mathop{\mathrm{poly}}\nolimits}

\newtheorem{observation}[theorem]{\textbf{Observation}}   
\newtheorem{assumption}{\textbf{Assumption}}

\title{Constructing Large Matchings via Query Access to a Maximal Matching Oracle}

\author{Lidiya Khalidah binti Khalil and Christian Konrad}

\institute{Department of Computer Science, University of Bristol, Bristol, UK \\ \texttt{\{lb17727,christian.konrad\}@bristol.ac.uk}}

\usepackage{float}
\usepackage{multicol}
\usepackage{wrapfig}
\usepackage{tikz}
\usepackage{xcolor}
\usepackage{sidecap}

\usetikzlibrary{decorations.pathreplacing}

\usepackage{subcaption}
\usepackage{caption}

\definecolor{lightred}{rgb}{1,0.8,0.8}

\newtheorem{problem}{Problem}

\begin{document}

 \maketitle
 
 \begin{abstract}
  Multi-pass streaming algorithm for \textsf{Maximum Matching} have been studied since more than 15 years and 
various algorithmic results are known today, including $2$-pass streaming algorithms that break the $1/2$-approximation barrier,
and $(1-\epsilon)$-approximation streaming algorithms that run in 
$\Order(\poly \frac{1}{\epsilon})$ passes in bipartite graphs and in $\Order( (\frac{1}{\epsilon})^{\frac{1}{\epsilon}})$ 
or $\Order(\poly (\frac{1}{\epsilon}) \cdot \log n)$ passes 
in general graphs, where $n$ is the number of vertices of the input graph. However, proving impossibility results for such algorithms has so far been elusive, and, for example, 
even the existence of $2$-pass small space streaming algorithms with approximation factor $0.999$ has not yet been ruled out.

The key building block of all multi-pass streaming algorithms for \textsf{Maximum Matching} is the \textsc{Greedy} matching algorithm.
Our aim is to understand the limitations of this approach: How many passes are required if the algorithm solely relies 
on the invocation of the \textsc{Greedy} algorithm?

In this paper, we initiate the study of lower bounds for restricted families of multi-pass streaming algorithms for \textsf{Maximum Matching}.
We focus on the simple yet powerful class of algorithms that in each pass run \textsc{Greedy} on a vertex-induced subgraph of the input graph.
In bipartite graphs, we show that $3$ passes are necessary and sufficient to improve on the trivial approximation factor of $1/2$: We give 
a lower bound of $0.6$ on the approximation ratio of such algorithms, which is optimal. We further show that $\Omega( \frac{1}{\epsilon})$ 
passes are required for computing a $(1-\epsilon)$-approximation, even in bipartite graphs. Last, the considered class of algorithms is not 
well-suited to general graphs: We show that $\Omega(n)$ passes are required in order to improve on the trivial approximation factor of $1/2$.
 \end{abstract}

\section{Introduction}
\label{sec:introduction}
The \textsc{Greedy} matching algorithm is the key building block of most published streaming algorithms for approximate 
\textsf{Maximum Matching} \cite{fkmsz05,m05,ekms12,kmm12,ag13,ehm16,kt17,k18}. Given a graph $G=(V, E)$, \textsc{Greedy} 
scans the set of edges $E$ in arbitrary order and inserts the current edge $e \in E$ into an initially empty matching 
$M$ if possible, i.e., if both endpoints of $e$ are not yet matched by an edge in $M$. \textsc{Greedy} produces
a maximal matching, which is known to be at least half as large as a matching of largest size. 

The \textsc{Greedy} matching algorithm is well-suited for implementation in the {\em streaming model of computation}. A streaming 
algorithm processing a graph $G=(V, E)$ with $|V| = n$ receives a potentially adversarially ordered sequence of the edges of the input 
graph, and the objective is to solve a graph problem using as little space as possible. Many graph problems require 
space $\Omega(n \log n)$ to be solved in the streaming model \cite{sw15}, and streaming algorithms that use space 
$\Order(n \poly \log n)$ are referred to as {\em semi-streaming algorithms}. 
Multi-pass streaming algorithms process the input stream multiple times. Observe that \textsc{Greedy} constitutes a 
one-pass semi-streaming algorithm for \textsf{Maximum Matching} with approximation factor $\frac{1}{2}$.

The \textsf{Maximum Matching} problem is the most studied graph problem in the streaming model, and despite intense research efforts,
the \textsc{Greedy} algorithm is the best one-pass streaming algorithm known today, even if space $\Order(n^{2-\delta})$ is 
allowed, for any $\delta > 0$. Performing multiple passes over the input allows improving the approximation factor.
The main questions addressed in the literature are: (1) What can be achieved in $p$ passes, for small $p$ (e.g. $p \in \{2, 3 \}$), 
and (2) How many passes
are required in order to obtain a $(1-\epsilon)$-approximation, for any $\epsilon > 0$. See Table~\ref{tab:known-results} for 
an overview of the currently best results.

\begin{table}[h]
\begin{center}
 \begin{tabular}{|llll|l|}
 \hline
 \# passes & Approximation & det/rand & Reference & See also \\
 \hline 
 \multicolumn{4}{|l|}{\textbf{Bipartite Graphs}} & $ $ \\    
  $1$ & $\frac{1}{2}$ & deterministic & \textsc{Greedy}, folklore & \\
  $2$ & $2-\sqrt{2} \approx 0.5857$ & randomized & Konrad \cite{k18} & \cite{kmm12,ehm16,kt17} \\
  $3$  & $0.6067$   & randomized & Konrad \cite{k18} & \cite{ehm16,kt17} \\
  $\frac{2}{3\epsilon}$ & $\frac{2}{3} - \epsilon$ & deterministic & Kale and Tirodkar \cite{kt17} & \cite{fkmsz05} \\
  $\Order(\frac{1}{\epsilon^2} \log \log \epsilon)$  & $1-\epsilon$ & deterministic & Ahn and Guha \cite{ag13} & \cite{ekms12} \\  
  \multicolumn{4}{|l|}{\textbf{General Graphs}} & $ $  \\
  $1$ & $\frac{1}{2}$ & deterministic & \textsc{Greedy}, folklore & \\
  $2$ & $0.53125$ & deterministic & Kale and Tirodkar  \cite{kt17} & \cite{kmm12} \\
  $\frac{1}{\epsilon}^{\Order(\frac{1}{\epsilon})}$ & $1-\epsilon$ & deterministic & Tirodkar \cite{t18} & \cite{m05}  \\
  $\Order(\frac{1}{\epsilon^4} \log n)$ & $1-\epsilon$ & deterministic & Ahn and Guha \cite{ag13} & $ $ \\
  \hline
 \end{tabular}
 
 \vspace{0.1cm}
 
 \caption{State of the art semi-streaming algorithms for \textsf{Maximum Matching}. \label{tab:known-results}}
\end{center}
 \end{table}

 Only few lower bounds are known: 
 We know that one-pass semi-streaming algorithms cannot have an approximation factor larger than $1-\frac{1}{e}$ \cite{k13} 
 (see also \cite{gkk12}). The only multi-pass lower bound known addresses the exact version of \textsf{Maximum Matching}, showing
 that computing a maximum matching in $p$ passes requires space $n^{1+\Omega(1/p)} / p^{O(1)}$ \cite{go16}. No lower bound 
 is known for multiple passes and approximations, and, for example, the existence of a $2$-pass $0.999$-approximation 
 semi-streaming algorithm has not yet been ruled out.
  
 The \textsc{Greedy} algorithm is the key building block of all algorithms referenced in Table~\ref{tab:known-results} (including 
 those mentioned in the ``See also'' column). In many cases, the presented algorithms collect edges by {\em solely} executing 
 \textsc{Greedy} on specific subgraphs in each pass and output a large matching computed from the edges produced by 
 \textsc{Greedy}. In this paper, we are interested in the limitations of this approach: How large a matching can be computed
 if \textsc{Greedy} is executed at most $p$ times?
 
 Known streaming algorithms apply \textsc{Greedy} in different ways. For example, the $2$-pass and $3$-pass algorithms
 by Konrad \cite{k18} run \textsc{Greedy} on randomly sampled subgraphs that depend on a previously computed maximal matching.
 The multi-pass algorithms by Ahn and Guha \cite{ag13} maintain vertex weights $\in [0, 1]$ over the course of the algorithm and run \textsc{Greedy}
 on a {\em threshold subgraph}, i.e., on the set of edges $uv$ so that the sum of the current weights associated with $u$ and $v$
 is at most $1$. The algorithm by Eggert et al. \cite{ekms12} runs \textsc{Greedy} on an edge-induced subgraph in order to find augmenting
 paths.
 
 In this paper, we initiate the study of lower bounds for restricted families of multi-pass streaming algorithms for 
 \textsf{Maximum Matching} that are based on \textsc{Greedy}. We start this line of research by addressing the 
 probably simplest and most natural approach, which is nevertheless surprisingly powerful:
 the class of deterministic algorithms that run \textsc{Greedy} on a vertex-induced subgraph in each pass.
 Two known streaming algorithms fit our model: 
 
 \begin{enumerate}
  \item A $3$-pass $0.6$-approximation streaming algorithm for bipartite graphs that is implicit in \cite{fkmsz05}, 
  explicitly mentioned in \cite{kmm12}, and analyzed in \cite{kt17}. Given a bipartite graph $G=(A, B, E)$, the algorithm
  first computes a maximal matching in $G$, i.e., $M \gets \textsc{Greedy}(G)$. Then, the algorithm attempts to find length-$3$ augmenting
  paths by invoking \textsc{Greedy} twice more: $M_L \gets \textsc{Greedy}(G[ A(M) \cup \overline{B(M)}])$, where $A(M)$ are the matched 
  $A$-vertices and $\overline{B(M)}$ are the unmatched $B$-vertices. Last, $M_R \gets \textsc{Greedy}(\overline{A(M)}, B')$, where
  $B' \subseteq B(M)$ are those matched $B$ vertices that are endpoints in length-$2$ paths in $M_L  \cup M$. Kale and Tirodkar 
  showed that $M \cup M_L \cup M_R$ contains a $0.6$-approximate matching \cite{kt17}. We will denote this algorithm by \textsc{3RoundMatch}.
  
  \item The $(1-\epsilon)$-approximation $\Order(\frac{1}{\epsilon^5})$-passes streaming algorithm for bipartite graphs 
  by Eggert et al. \cite{ekms12} can be adapted to fit our model using $\Order(\frac{1}{\epsilon^6})$ invocations of \textsc{Greedy}. 
 \end{enumerate}
 
 We abstract this approach as a game
 between a player and an oracle: Let $G$ be a graph with vertex set $V$. The player initially knows $V$. In each round $i$
 the player sends a query \textsf{query}($V_i$) to the oracle, where $V_i \subseteq V$. The oracle returns a maximal matching
 in the vertex-induced subgraph $G[V_i]$. For this model to yield lower bounds for the streaming model, we impose that the oracle
 is {\em streaming-consistent}, i.e., there exists a stream of edges $\pi$ so that the oracle's answers to 
 the queries $(\textsf{query}(V_i))_i$ equal runs of \textsc{Greedy} on the respective substream of edges $G[V_i]$ of $\pi$ (see 
 preliminaries for a more detailed definition). We denote this model as the {\em vertex-query model} (as opposed to an
 edge-query model, where the player may ask for maximal matchings in a subgraph spanned by a subset of edges).

 \begin{figure}[h]
\centering

\begin{tikzpicture}
\node [circle, draw, minimum size=1cm, fill=lightred] (player) at (-3,0) {Player};
\node [circle, draw, minimum size=1cm, fill=yellow] (oracle) at (3,0) {Oracle};

\path[->]
	(player) edge [bend left=25, above] node {\textsf{query}($V_r$)} (oracle)
	(oracle) edge [bend left=25, below] node {response: maximal matching in $G[V_r]$} (player);
\end{tikzpicture}
\caption{Illustration of the game between the player and oracle in the vertex-query model.}
\end{figure}
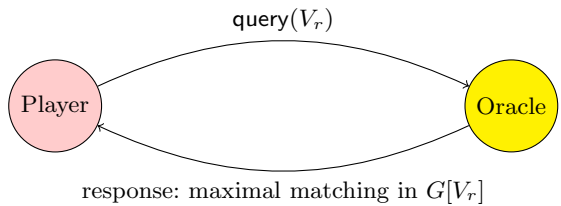

\vspace{0.1cm} \noindent \textbf{Our Results.} 
In bipartite graphs, we show that at least $3$ rounds are required to improve on the approximation factor 
 of $1/2$, and we give a lower bound of $0.6$ on the approximation factor of $3$ round algorithms. This is optimal, as demonstrated by 
 the previously mentioned algorithm \textsc{3RoundMatch}. We also show that $\Omega(\frac{1}{\epsilon})$ rounds are required
 for computing a $(1-\epsilon)$-approximation. This polynomial lower bound is in line with the $\poly \frac{1}{\epsilon}$ rounds 
 upper bound by Eggert et al. \cite{ekms12}. Last, we demonstrate that our query model
 is not well-suited to general graphs: We show that improving on a factor of $1/2$ requires $\Omega(n)$ rounds.

\vspace{0.1cm} \noindent \textbf{Further Related Work.}
Besides the adversarial one-pass and multi-pass streaming models, \textsf{Maximum Matching} has also been
studied in the random order \cite{kmm12, k18, gkms19, abbms19, fhmrr20, b20} and the insertion-deletion settings \cite{k15,ccehmmv16,akly16,dk20}. 
In the random order model, where edges arrive in uniform random order,
Konrad et al. \cite{kmm12} were the first to give a semi-streaming algorithms with approximation ratio above $1/2$. 
Very recently, Bernstein showed that an approximation ratio of $2/3$ can be achieved in random order streams \cite{b20}. In light
of the lower bound by Kapralov \cite{k13}, this result separates the adversarial and the random order settings.
In insertion-deletion streams, edges that have previously been inserted may be deleted again. Assadi et al. \cite{akly16} showed that,
up to sub-polynomial factors, space $n^{2-3\epsilon}$ is necessary and sufficient for computing a $n^\epsilon$-approximation (see \cite{dk20} for a slightly improved lower bound).

Many works allow only query access to the input graph. For example, cross-additive queries, bipartite independent set queries, additive queries,
cut-queries, and edge-detection queries have been considered \cite{gk00,abkrs04,ck08,s13,bhrrs18,rsw18,ab19}, however, mainly for graph reconstruction problems. Very recently, linear queries and or-queries have been considered for graph connectivity \cite{ack20}.

\vspace{0.1cm} \noindent \textbf{Outline.} In Section~\ref{sec:prelim}, we give notation and definitions. We also define the vertex-query model
and provide a construction mechanism that ensures that our oracles are streaming-consistent. Then, in Section~\ref{sec:few-rounds} we 
prove that $3$ rounds are required to improve on $1/2$ and give a lower bound of $0.6$ on the approximation
ratio achievable in three rounds. In Section~\ref{sec:multi-rounds}, we show that $\Omega(\frac{1}{\epsilon})$ rounds
are needed for computing a $(1-\epsilon)$-approximation, and in Section~\ref{sec:general-graphs} we show that improving
on $\frac{1}{2}$ in general graphs requires $\Omega(n)$ rounds. Finally, we conclude in Section~\ref{sec:conclusion} and give open questions.

\section{Preliminaries}\label{sec:prelim}

\vspace{0.1cm} \noindent \textbf{Matchings.} Let $G=(V, E)$ be a graph with $|V| = n$. A {\em matching} $M \subseteq E$ is a subset of vertex-disjoint edges.
Matching $M$ is {\em maximal} if for every $e \in E \setminus M: M \cup \{e\}$ is not a matching. A {\em maximum matching} 
is one of largest cardinality. If the size of a matching $M$ is $n/2$, i.e., it matches all vertices of the graph, then 
$M$ is a {\em perfect matching}.

\vspace{0.1cm} \noindent \textbf{Notation.} We write $V(M)$ to denote the set of vertices incident to the edges of a matching $M$.
For a subset of vertices $V' \subseteq V$, we denote by $G[V']$ the {\em vertex-induced subgraph} of $G$ by vertices $V'$, i.e.,
$G[V'] = (V', (V' \times V') \cap E)$. For a set of edges $E' \subseteq E$, we denote by $OPT(E')$ the size of a maximum matching
in the subgraph of $G$ spanned by the edges $E'$. For an integer $n$, we define $[n] := \{1, 2, \dots, n\}$.

\vspace{0.1cm} \noindent \textbf{The Vertex-query Model.} In the \emph{vertex-query model}, a player and an oracle
play a rounds-based matching game on a vertex set $V$ of size $n$ that is initially known to both parties.  
Over the course of the game, the oracle makes up a graph $G=(V, E)$. The objective of the player is
to learn a large matching in $G$. The way the player learns edges is as follows:

In each round $1 \le i \le r$, where $r$ is the total number of rounds played, 
the player submits a query $\textsf{query}(V_i)$ to the oracle, for some $V_i \subseteq V$. The oracle 
then determines a set of edges $M_i$, which is guaranteed to be a maximal matching in the vertex-induced
subgraph $G[V_i]$. Observe that in doing so, the oracle not only commits to the fact that $M_i \subseteq E$,
but also that the vertices $V_i \setminus V(M_i)$ form an independent set (which follows from the fact that
$M_i$ is maximal). Furthermore, we impose that the answers to all queries are consistent with graph $G$
and that $G$ has a perfect matching. 

After the $r$ query rounds, the player reports a largest matching $M_P$ 
that can be formed using the edges $\cup_{i \le r} M_i$. The {\em approximation ratio} of the solution 
obtained is $|M_P| / (\frac{1}{2} n)$.


We are interested in oracles that are consistent with the streaming model. We say that an oracle is {\em streaming-consistent}, 
if there exists an ordering $\pi$ of the edges $E$ so that, for every round $i$, $M_i$ is produced by 
running \textsc{Greedy} on the substream of $\pi$ consisting of the edges of $G[V_i]$. We will ensure that all our oracles
are streaming-consistent.

\vspace{0.1cm} \noindent \textbf{Construction of Streaming-consistent Oracles.}
We will construct streaming-consistent oracles as follows. Upon query $V_1$, the oracle answers with $M_1$ and
places $M_1$ in the beginning of the stream $\pi$. Next, given query $V_i$, for some $i \ge 2$, the oracle first runs 
\textsc{Greedy} on the substream of $\pi$ consisting of the edges $G[V_i]$ which produces an intermediate matching $M'$, 
thereby attempting to match $V_i$ using edges of previous matchings $\cup_{j < i} M_j$. The oracle then extends $M'$
to a matching $M_i$. Edges $M_i \setminus M'$ are then introduced at the end of the stream $\pi$. 
This construction procedure guarantees that our oracles are streaming-consistent. Furthermore, it allows us to simplify 
our arguments, since it is enough to restrict our considerations to queries with the following property:

\begin{observation}\label{obs:queries}
 Suppose that the oracle is constructed as above. Then, given the sequence of queries $V_1, \dots, V_r$ and matchings
 $M_1, \dots, M_r$, there exists a sequence of queries $\tilde{V}_1, \dots, \tilde{V}_r$ that produces matchings 
 $\tilde{M}_1, \dots, \tilde{M}_r$ such that:
 \begin{itemize}
  \item The player learns the same set of edges, i.e., for every $i \le r: \bigcup_{j \le i} M_j = \bigcup_{j \le i} \tilde{M}_j$, and
  \item No query $\tilde{V}_i$ contains a pair of vertices $u,v$ such that $uv \in \cup_{j < i} \tilde{M}_j$. 
 \end{itemize}
\end{observation}

We can therefore assume that the player never includes a pair of vertices $u,v$ into a query so that the edge $uv$ is contained
in a previous answer from the oracle.

\section{Lower Bound for Few Round Algorithms in Bipartite Graphs} \label{sec:few-rounds}

In this section, we show that the player cannot produce an approximation ratio better than $\frac{1}{2}$ in two rounds, even on 
bipartite graphs. We also show that three rounds do not allow for an approximation ratio better than $0.6$, which is achieved
by the algorithm \textsc{3RoundMatching}.

In order to keep track of the information 
learned by the player, we will make use of {\em structure graphs}, which we discuss first. 

\subsection{Structure Graphs}
Observe that when the oracle answers the query $\textsf{query}(V_i)$ and returns a maximal matching $M_i$, 
the player not only learns that the edges $M_i$ are contained in the input graph $G$, but also learns that
the vertices $V_i \setminus V(M_i)$ form an independent set in $G$ (due to the maximality of $M_i$). 
We maintain the structure learned by the player and the structure committed to by the oracle (which do not have to be 
identical) using {\em structure graphs}:

\begin{definition}[Structure graph]
A 4-tuple $(A, B, E, F)$ is a \textbf{bipartite structure graph} if:
\begin{itemize}
\item $A, B$ are disjoint sets of vertices,
\item $E, F$ are disjoint sets of edges such that $(A, B, E)$ and $(A, B, F)$ are bipartite graphs,
\item The structure graph admits a perfect matching, i.e., there exists a set of edges $M^*$ such that 
$M^* \cap F = \emptyset$ and $M^*$ is a perfect matching in the bipartite graph $(A, B, E \cup M^*) \ .$
\end{itemize}
\end{definition}

From the perspective of the player, the set $E$ corresponds to the edges returned by the oracle so far, i.e., $E = \cup_{j \le i} M_j$,
and the set $F$ corresponds to guaranteed {\em non-edges}, i.e., $F = \cup_{j \le i} C(V_i \setminus V(M_i))$, where $C(V')$ denotes
a biclique (respecting the bipartition $A,B$) among the vertices $V'$. 

In the following, we will denote the structure graph after round $i$ learned by the player by 
$\tilde{H}_i = (A, B, \tilde{E}_i, \tilde{F}_i)$, i.e., $\tilde{E}_i = \cup_{j \le i} M_j$ and 
$\tilde{F}_i = \cup_{j \le i} C(V_i \setminus V(M_i))$. The oracle will also maintain a sequence of structure graphs 
$(H_i)_i$ with $H_i = (A, B, E_i, F_i)$ such that $H_i$ {\em dominates} $\tilde{H}_i$, for every $ 1\le i \le r$. We say that 
a structure graph $H=(A, B, E, F)$ dominates a structure graph $\tilde{H} = (A, B, \tilde{E}, \tilde{F})$, if 
$\tilde{E} \subseteq E$ and $\tilde{F} \subseteq F$. This notion allows the oracle to commit to edges and non-edges
that the player has not yet learned. This domination property allows us to simplify our arguments. 

In our lower bound arguments, we make use of the following two assumptions:

\begin{assumption} \label{ass:1}
 After round $i$, the player knows the structure graph $H_i$.
\end{assumption}
This is a valid assumption since $H_i$ dominates $\tilde{H}_i$ and thus contains at least as much information as 
$\tilde{H}_i$. This assumption therefore 
only strengthens the player. Furthermore, we will also assume a slightly strengthened property of the property 
discussed in Observation~\ref{obs:queries}:

\begin{assumption} \label{ass:2}
 For every $1 \le i \le r$, we assume that query $V_i$ does not contain a pair of vertices $u,v \in V_i$ such that $uv \in E_{i-1}$.
\end{assumption}

This is a valid assumption, since if such a pair $u,v$ of vertices existed in $V_i$, the oracle could simply match $u$ to 
$v$ in $M_i$ and the algorithm would not learn any new information.

Last, observe that the approximation ratio of the player's strategy is completely determined by $H_r$, the oracle's structure 
graph after the last round. Since $H_r$ dominates $\tilde{H}_r$, the player's largest matching is 
of size at most $OPT(E_r)$. Since by definition of a structure graph, $H_r$ admits a perfect matching,
the approximation ratio achieved is $2 \cdot OPT(E_r) / n$.

\subsection{Lower Bound for Two Rounds}
Assume that $n$ is a multiple of $4$. The player and the oracle play the matching game on a bipartite vertex set $V = A \ \dot{\cup} \ B$ with 
$|A| = |B| = n/2$. Consider the structure graph:
$$H_1 = (A_{in} \cup A_{out}, B_{in} \cup B_{out}, M, A_{out} \times B_{out}) \ ,$$
where $|A_{in}| = |A_{out}| = |B_{in}| = |B_{out}| = n / 4$, and $M$ is a perfect matching between $A_{in}$ and $B_{in}$.
Observe that there exists an $M^*$ outside $A_{out} \times B_{out}$ such that $M^*$ is a perfect matching in $(A, B, M \cup M^*)$, 
namely, $M^*$ consists of the two perfect matchings $M_L^*$ connecting $B_{out}$ to $A_{in}$ and $M_R^*$ connecting 
$B_{in}$ to $A_{out}$. See Figure~\ref{fig:structure-graph} for an illustration.

\begin{figure}[!t]
\begin{center}
\begin{minipage}{0.47 \textwidth}
 \begin{figure}[H]
\centering
\begin{tikzpicture}[scale=0.6]
\tikzstyle{vertex} = [circle, fill=black, inner sep=0pt, minimum size=0.25cm]
\tikzstyle{black line} = [line width=1.3pt, color=black]
\tikzstyle{blue dash} = [color=blue, line width=1.3pt, dashed]
 
\node [vertex][label=above:$B_{out}$] (1) at (-3.5,3.5) {};
\node [vertex][label=above:$A_{in}$] (2) at (-1.5,3.5) {};
\node [vertex][label=above:$B_{in}$] (3) at (0.5,3.5) {};
\node [vertex][label=above:$A_{out}$] (4) at (2.5,3.5) {};
\node at (-0.5,-1.8) {$M$};
\node at (-2.5,-1.8) {$M_L^*$};
\node at (1.5,-1.8) {$M_R^*$};
\node [vertex] (5) at (-3.5,2) {};
\node [vertex] (6) at (-1.5,2) {};
\node [vertex] (7) at (0.5,2) {};
\node [vertex] (8) at (2.5,2) {};
\node [vertex] (9) at (-3.5,0.5) {};
\node [vertex] (10) at (-1.5,0.5) {};
\node [vertex] (11) at (0.5,0.5) {};
\node [vertex] (12) at (2.5,0.5) {};
\node [vertex] (13) at (-3.5,-1) {};
\node [vertex] (14) at (-1.5,-1) {};
\node [vertex] (15) at (0.5,-1) {};
\node [vertex] (16) at (2.5,-1) {};

\draw [black line] (2) edge (3);
\draw [black line] (6) edge (7);
\draw [black line] (10) edge (11);
\draw [black line] (14) edge (15);

\draw [blue dash] (1) edge (2);
\draw [blue dash] (3) edge (4);
\draw [blue dash] (5) edge (6);
\draw [blue dash] (7) edge (8);
\draw [blue dash] (9) edge (10);
\draw [blue dash] (11) edge (12);
\draw [blue dash] (13) edge (14);
\draw [blue dash] (15) edge (16);
\end{tikzpicture}
\caption{Illustration of the structure graph $H_1$ on a graph on $16$ vertices. The matching $M$ is half the size of the matching $M^* = M^*_L \cup M^*_R$.} 
\label{fig:structure-graph}
\end{figure}
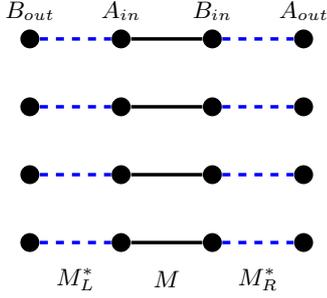
\end{minipage} \hfill
\begin{minipage}{0.47 \textwidth}
\begin{figure}[H]
\centering
\begin{tikzpicture}[scale=0.6]
\tikzstyle{vertex} = [circle, fill=black, inner sep=0pt, minimum size=0.25cm]
\tikzstyle{vertex2} = [circle, fill=red, inner sep=0pt, minimum size=0.25cm]
\tikzstyle{black line} = [line width=1.0pt, color=black]
\tikzstyle{red line} = [line width=1.5pt, color=red]
\tikzstyle{blue dash} = [color=blue, line width=1.0pt, dashed]
 
\node [vertex2][label=above:$B_{out}$] (1) at (-3.5,3.5) {};
\node [vertex2][label=above:$A_{in}$] (2) at (-1.5,3.5) {};
\node [vertex][label=above:$B_{in}$] (3) at (0.5,3.5) {};
\node [vertex][label=above:$A_{out}$] (4) at (2.5,3.5) {};
\node at (-0.5,-1.8) {$M$};
\node at (-2.5,-1.8) {$M_L^*$};
\node at (1.5,-1.8) {$M_R^*$};
\node [vertex] (5) at (-3.5,2) {};
\node [vertex2] (6) at (-1.5,2) {};
\node [vertex] (7) at (0.5,2) {};
\node [vertex] (8) at (2.5,2) {};
\node [vertex2] (9) at (-3.5,0.5) {};
\node [vertex2] (10) at (-1.5,0.5) {};5
\node [vertex] (11) at (0.5,0.5) {};
\node [vertex2] (12) at (2.5,0.5) {};
\node [vertex] (13) at (-3.5,-1) {};
\node [vertex] (14) at (-1.5,-1) {};
\node [vertex2] (15) at (0.5,-1) {};
\node [vertex2] (16) at (2.5,-1) {};

\draw [black line] (2) edge (3);
\draw [black line] (6) edge (7);
\draw [black line] (10) edge (11);
\draw [black line] (14) edge (15);

\draw [blue dash] (1) edge (2);
\draw [blue dash] (3) edge (4);
\draw [blue dash] (5) edge (6);
\draw [blue dash] (7) edge (8);
\draw [blue dash] (9) edge (10);
\draw [blue dash] (11) edge (12);
\draw [blue dash] (13) edge (14);
\draw [blue dash] (15) edge (16);

\draw [red line] (10) edge (15);
\draw [red line] (1) edge (2);
\draw [red line] (6) edge (9);
\end{tikzpicture}
\caption{Matching $M_2$ (in red) returned by the oracle. The red vertices constitute $A_2 \cup B_2$, i.e., the vertices of the second query. 
The case $|A_2^{in}| \ge |B_2^{in}|$ is illustrated here. We see that no edges from $B_{in} \times A_{out}$ are returned, and 
that $M_2$ does not allow us to increase the size of $M$. \label{fig:two-rounds}}
\end{figure}
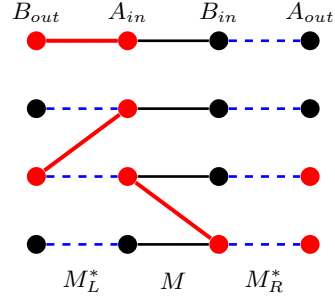
\end{minipage}
\end{center}
\end{figure}
We have: 

\begin{lemma}\label{lem:domination}
 There is a structure graph isomorphic to $H_1$ that dominates $\tilde{H}_1$.
\end{lemma}
\begin{proof}
 Denote the first query by $A_1, B_1$ ($A_1 \subseteq A$, and $B_1 \subseteq B$). We will argue that we 
 can relabel the sets $A_{in}, A_{out}, B_{in}, B_{out}$ so that 
 $H_1$ dominates $\tilde{H}_1$:
 
 If $A_1 \le n/4$ then let $A_{in}$ be an arbitrary subset of the $A$ vertices of size $n/4$ that 
 contains $A_1$, and let $A_{out}$ be the remaining $A$-vertices. If $A_1 > n/4$ then let $A_{out}$ be an arbitrary subset of $A$ vertices
 of size $n/4$ that contains $A \setminus A_1$, and let $A_{in}$ be the remaining $A$-vertices. Proceed similarly for $B_1$. 
 The oracle returns the subset $M_1 \subseteq M$ where each edge has one endpoint in $A_1$ and one endpoint in $B_1$, which is
 clearly maximal given that edges in $A_{out} \times B_{out}$ are forbidden.
\end{proof}

Since $OPT(M) = |M| = \frac{1}{4}n$, Lemma~\ref{lem:domination} implies the unsurprising fact that no one round algorithm
has an approximation ratio better than $\frac{2 \cdot \frac{1}{4}n}{n} = \frac{1}{2}$. We argue now that an additional 
round does not help with increasing the approximation factor. 

\begin{theorem}
 The best approximation ratio achievable in two rounds is $1/2$.
\end{theorem}
\begin{proof}
 Let $A_2, B_2$ be the vertices of the second query. By Lemma~\ref{lem:domination}, 
 $H_1$ dominates $\tilde{H}_1$, and by Assumption~\ref{ass:1} we can assume that the player already knows $H_1$. 
 Let $A_2^{in} = A_2 \cap A_{in}$, $A_2^{out} = A_2 \cap A_{out}$ and define $B_2^{in}$ and $B_2^{out}$ similarly. 
 
 Suppose first
 that $|A_2^{in}| \ge |B_2^{in}|$. Then the oracle returns a matching $M_2$ that matches an arbitrary subset of $A_2^{in}$ of size
 $|B_2^{in}|$ to $B_2^{in}$, and matches $\max \{|B_2^{out}|, |A_2^{in}| - |B_2^{in}|\}$ of the remaining 
 $A_2^{in}$ vertices arbitrarily to vertices in $B_2^{out}$. In doing so, either all $A_2^{in}$ vertices 
 or all $B_2$ vertices are matched. Since $H_1$ indicates that there are no edges connecting the ``out''-vertices, 
 $M_2$ is therefore maximal.
 
 Observe further that $M \cup M_2$ does not match any vertex in $A_{out}$, and, hence, only half of the $A$-vertices
 are matched in $M \cup M_2$.
 The player thus cannot report any matching of size larger than $|M|$, which constitutes a $1/2$-approximation. 
 
 Last, the case $|A_2^{in}| < |B_2^{in}|$ is identical with roles of $A$ and $B$ vertices reversed.
\end{proof}

\subsection{Lower Bound for Three Rounds}
In this section, we work with a vertex set $V = A \ \dot{\cup} \ B$ with $|A| = |B| = 5$ (and thus $|V| = n = 10$). 
By choosing disjoint copies of this vertex set, our result can be extended to graphs with an arbitrarily large number of vertices.

\vspace{0.1cm} \noindent \textbf{First Query.}
Similar to the two round case, we define the structure graph 
$H_1 = (A_{in} \cup A_{out}, B_{in} \cup B_{out}, M, A_{out} \times B_{out})$, however, this time 
$|A_{in}| = |B_{in}| = 3$ and $|A_{out}| = |B_{out}| = 2$. The matching $M$ matches $A_{in}$ to $B_{in}$, see Figure~\ref{fig:str-graph}.

It shall be convenient to assign labels to the vertices in our structure graph. In
our arguments below, in order to avoid symmetric cases, we relabel the vertices of our structure graph as we see fit, 
however, we always ensure that the structure graph after relabeling is isomorphic to the structure graph before the relabeling. 

First, similar to Lemma~\ref{lem:domination}, it is not hard to see that a structure graph isomorphic to $H_1$ dominates $\tilde{H}_1$ (proof omitted).

\begin{lemma}
 There is a structure graph isomorphic to $H_1$ that dominates $\tilde{H}_1$.
\end{lemma}

\vspace{0.1cm} \noindent \textbf{Second Query.}
We assume that the player knows $H_1$ after the first query (Assumption~1).
Next, we define structure graph $H_2 = (A_{in} \cup A_{out}, B_{in} \cup B_{out}, M \cup E_2, A_{out} \times B_{out} \cup F_2) $, where
$E_2 = \{ a_1b_5, a_2b_3 \}$, and $F_2 = \{ a_2b_4, a_3b_4 \}$. It is easy to see
that $H_2$ is indeed a structure graph (see Figures~\ref{fig:h2-1} and \ref{fig:h2-2}). 

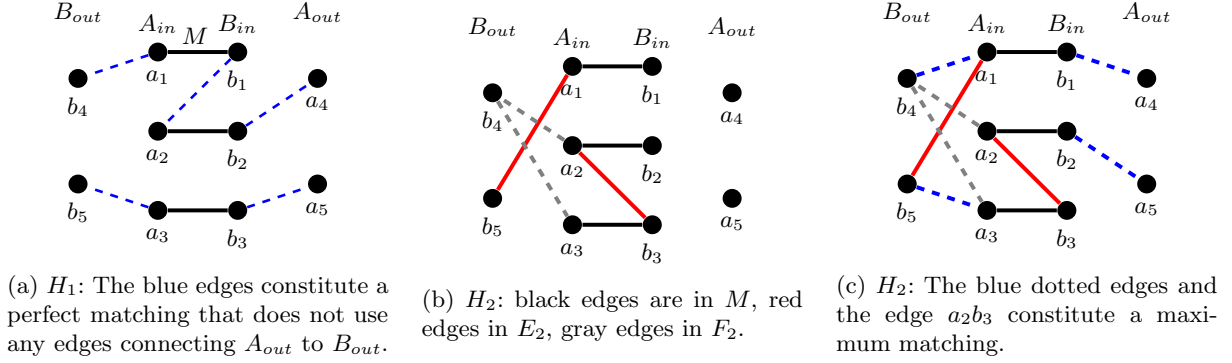
\begin{figure}[t]
\begin{subfigure}{0.315\textwidth}
\centering
\begin{tikzpicture}[scale=0.7]
\tikzstyle{vertex} = [circle, fill=black, inner sep=0pt, minimum size=0.25cm]
\tikzstyle{black line} = [line width=1.5pt, color=black]
\tikzstyle{blue dash} = [color=blue, line width=1.0pt, dashed]

\node [vertex][label={[label distance=0.5cm]above:$B_{out}$}, label=below:$b_4$] (v8) at (-3.5,3) {};
\node [vertex][label=above:$A_{in}$, label=below:$a_1$] (v1) at (-2,3.5) {};
\node [vertex][label=above:$B_{in}$, label=below:$b_1$] (v2) at (-0.5,3.5) {};
\node [vertex][label=below:$a_2$] (v3) at (-2,2) {};
\node [vertex][label=below:$b_2$] (v4) at (-0.5,2) {};
\node [vertex][label={[label distance=0.5cm]above:$A_{out}$}, label=below:$a_4$] (v9) at (1,3) {};
\node [vertex][label=below:$b_5$] (v7) at (-3.5,1) {};
\node [vertex][label=below:$a_3$] (v5) at (-2,0.5) {};
\node [vertex][label=below:$b_3$] (v6) at (-0.5,0.5) {};
\node [vertex][label=below:$a_5$] (v10) at (1,1) {};
\draw [black line] (v1) edge (v2);
\draw [black line] (v3) edge (v4);
\draw [black line] (v5) edge (v6);
\draw [blue dash] (v1) edge (v8);
\draw [blue dash] (v5) edge (v7);
\draw [blue dash] (v4) edge (v9);
\draw [blue dash] (v6) edge (v10);
\draw [blue dash] (v2) edge (v3);
\node at (-1.3,3.8) {$M$};
\end{tikzpicture}
\caption{$H_1$: The blue edges constitute a perfect matching that does not 
use any edges connecting $A_{out}$ to $B_{out}$. 
\label{fig:str-graph}}
\end{subfigure} \hfill
\begin{subfigure}{0.315\textwidth}
\centering

\begin{tikzpicture}[scale=0.7]
\tikzstyle{vertex} = [circle, fill=black, inner sep=0pt, minimum size=0.25cm]
\tikzstyle{green vertex} = [circle, fill=green, inner sep=0pt, minimum size=0.25cm]
\tikzstyle{blue dash} = [color=blue, line width=1.5pt, dashed]
\tikzstyle{gray dash} = [color=gray, line width=1.5pt, dashed]
\tikzstyle{red line} = [line width=1.5pt, color=red]
\tikzstyle{orange line} = [line width=1.5pt, color=orange]
\tikzstyle{black line} = [line width=1.5pt, color=black]

\node [vertex][label={[label distance=0.5cm]above:$B_{out}$}, label=below:$b_4$] (v8) at (-3.5,3) {};
\node [vertex][label=above:$A_{in}$, label=below:$a_1$] (v1) at (-2,3.5) {};
\node [vertex][label=above:$B_{in}$, label=below:$b_1$] (v2) at (-0.5,3.5) {};
\node [vertex][label=below:$a_2$] (v3) at (-2,2) {};
\node [vertex][label=below:$b_2$] (v4) at (-0.5,2) {};
\node [vertex][label={[label distance=0.5cm]above:$A_{out}$}, label=below:$a_4$] (v9) at (1,3) {};
\node [vertex][label=below:$b_5$] (v7) at (-3.5,1) {};
\node [vertex][label=below:$a_3$] (v5) at (-2,0.5) {};
\node [vertex][label=below:$b_3$] (v6) at (-0.5,0.5) {};
\node [vertex][label=below:$a_5$] (v10) at (1,1) {};
\draw [black line] (v1) edge (v2);
\draw [black line] (v3) edge (v4);
\draw [black line] (v5) edge (v6);
\draw [red line] (v1) edge (v7);
\draw [red line] (v3) edge (v6);
\draw [gray dash] (v8) edge (v5);
\draw [gray dash] (v3) edge (v8);
\end{tikzpicture}
\caption{$H_2$: black edges are in $M$, red edges in $E_2$, gray edges in $F_2$. \\  $ $ \label{fig:h2-1}}
\end{subfigure} \hfill
\begin{subfigure}{0.315\textwidth}
\centering

\begin{tikzpicture}[scale=0.7]
\tikzstyle{vertex} = [circle, fill=black, inner sep=0pt, minimum size=0.25cm]
\tikzstyle{green vertex} = [circle, fill=green, inner sep=0pt, minimum size=0.25cm]
\tikzstyle{blue dash} = [color=blue, line width=1.5pt, dashed]
\tikzstyle{gray dash} = [color=gray, line width=1.5pt, dashed]
\tikzstyle{red line} = [line width=1.5pt, color=red]
\tikzstyle{orange line} = [line width=1.5pt, color=orange]
\tikzstyle{black line} = [line width=1.5pt, color=black]

\node [vertex][label={[label distance=0.5cm]above:$B_{out}$}, label=below:$b_4$] (v8) at (-3.5,3) {};
\node [vertex][label=above:$A_{in}$, label=below:$a_1$] (v1) at (-2,3.5) {};
\node [vertex][label=above:$B_{in}$, label=below:$b_1$] (v2) at (-0.5,3.5) {};
\node [vertex][label=below:$a_2$] (v3) at (-2,2) {};
\node [vertex][label=below:$b_2$] (v4) at (-0.5,2) {};
\node [vertex][label={[label distance=0.5cm]above:$A_{out}$}, label=below:$a_4$] (v9) at (1,3) {};
\node [vertex][label=below:$b_5$] (v7) at (-3.5,1) {};
\node [vertex][label=below:$a_3$] (v5) at (-2,0.5) {};
\node [vertex][label=below:$b_3$] (v6) at (-0.5,0.5) {};
\node [vertex][label=below:$a_5$] (v10) at (1,1) {};
\draw [black line] (v1) edge (v2);
\draw [black line] (v3) edge (v4);
\draw [black line] (v5) edge (v6);
\draw [red line] (v1) edge (v7);
\draw [red line] (v3) edge (v6);
\draw [gray dash] (v8) edge (v5);
\draw [gray dash] (v3) edge (v8);
\draw [blue dash] (v8) edge (v1);
\draw [blue dash] (v2) edge (v9);
\draw [blue dash] (v7) edge (v5);
\draw [blue dash] (v4) edge (v10);
\end{tikzpicture}
\caption{$H_2$: The blue dotted edges and the edge $a_2b_3$ constitute a maximum matching. \\ $ $  \label{fig:h2-2}}
\end{subfigure}
\caption{Illustrations of structure graphs $H_1$ and $H_2$.  \label{fig:h2}} 
\end{figure}

We shall prove that there is a structure graph isomorphic to $H_2$ that dominates $\tilde{H}_2$. 
Lemma~\ref{lem: 2rnd 3in} considers the case when the second query $V_2$ contains exactly three ``in''-vertices, i.e., 
vertices from $A_{in} \cup B_{in}$, and Lemma~\ref{lem: 2rnd <3} considers the case when there are fewer ``in''-vertices.
By Assumption~\ref{ass:2}, we do not need to consider the cases when more than three ``in''-vertices are
contained in $V_2$ since then $V_2$ necessarily contains a pair of vertices $u,v$ such that $uv \in M$.

\begin{lemma}
\label{lem: 2rnd 3in}
If the player queries exactly 3 ``in''-vertices (i.e., vertices from $A_{in} \cup B_{in}$) in their second query then there exists 
a structure graph isomorphic to $H_2$ that dominates $\tilde{H}_2$.
\end{lemma}
\begin{proof}
The player can either query more vertices in $A_{in}$ or in $B_{in}$, and these cases are symmetrical. Hence we 
only consider the case when the player queries more vertices in $A_{in}$. Due to Assumption~\ref{ass:2}, for 
queries that contain vertices in both $A_{in}$ and $B_{in}$, we assume these vertices do not form any edges seen in $M$.


Since we will not match any vertices in $A_{out}$, we do not need to distinguish between 
cases where the player queries different numbers of vertices in $A_{out}$. We distinguish between the following cases:

\begin{enumerate}
	\item Player queries all vertices in $A_{in}$ and the query includes $b_5$: the oracle returns  $M_2 = \{ a_1b_5 \}$.
	\item Player queries all vertices in $A_{in}$ and only $b_4$ in $B_{out}$: relabel $b_4$ as $b_5$ and proceed as in case (1). 
	\item Player queries all vertices in $A_{in}$ and no vertices in $B_{out}$: the oracle returns $M_2 = \emptyset$.
	\item Player queries two vertices in $A_{in}$, one vertex in $B_{in}$ and the query includes $b_5$: relabel the ``in'' 
	vertices so that after relabeling the vertices $a_1$, $a_2$ and $b_3$ are included in the query. The oracle returns $M_2 = E_2$.
	\item Player queries two vertices in $A_{in}$, one vertex in $B_{in}$ and only $b_4$ in $B_{out}$: relabel $b_4$ as $b_5$ and proceed as in case (4).
	\item Player queries two vertices in $A_{in}$, one vertex in $B_{in}$ and no vertices in $B_{out}$:  relabel ``in'' vertices so that after relabeling the vertices $a_2$ and $b_3$ are included in the query. The oracle returns $M_2 = \{ a_2b_3 \}$.
\end{enumerate}
In all cases considered, observe that $M_2 \subseteq E_2$. Further, edges $F_2$ ensure that $M_2$ is maximal.
\end{proof}

We argue now that querying three ``in''-vertices in the second round is best possible in the sense
that querying fewer (or more) ``in''-vertices does not yield more information.

\begin{lemma}
\label{lem: 2rnd <3}
If the player queries fewer than 3 ``in''-vertices (i.e., vertices from $A_{in} \cup B_{in}$) then there exists 
a structure graph isomorphic to $H_2$ that dominates $\tilde{H}_2$.
\end{lemma}
\begin{proof}
Clearly if the player does not query any ``in''-vertices, no matching will be found i.e. $M_2 = \emptyset$. 
If the player queries exactly one vertex in $A_{in}$, we can relabel this vertex as $a_1$ and if the query contains a
vertex in $B_{out}$, relabel this one to be $b_5$. Then the matching found will be a subset of $E_2$. 
If the player queries exactly two ``in'' vertices there are two cases to consider. If they 
are both in $A_{in}$, we ensure one of these vertices is $a_1$ by relabeling, and, if at least one vertex in 
$B_{out}$ is queried, potentially relabel this vertex to be $b_5$ and return the edge $a_1b_5$. If the player 
queried one vertex in $A_{in}$ and one in $B_{in}$, we relabel these vertices as $a_2, b_3$ and return the edge 
between them, $a_2b_3$. Hence the edges learned by the player are always a subset of $E_2$.
In all cases considered, edges $F_2$ ensure that matching $M_2$ is maximal.
\end{proof}



\vspace{0.1cm} \noindent \textbf{Third Query.}
We assume that the player knows structure graph $H_2$. Similar to the second query, we distinguish between the cases
where the player queries exactly three ``in''-vertices and fewer ``in''-vertices. Again, by Assumption~\ref{ass:2},
we do not need to consider the case where the player queries more than three ``in''-vertices.
In the following proofs, we will define different structure graphs $H_3$ that depend on the individual query.

\begin{lemma} \label{lem:exactly-3-in-vertices}
If the player queries exactly 3 ``in''-vertices in the third round, then the player cannot output a matching of size larger than $3$.
\end{lemma}

\begin{proof}
We provide the oracle's answers when the player queries exactly three ``in''-vertices. 
Among those cases, there are three cases to consider where the player queries more vertices in $B_{in}$ than in $A_{in}$: 
\begin{enumerate}
 \item \textbf{Case 1:} Player queries $b_1, b_2, b_3$.
The oracle defines $H_3 = (A, B, E_3, F_3)$ such that $E_3 = M \cup E_2 \cup \{ a_4b_2, a_5b_3 \}$ and $F_3 = A_{out} \times B_{out} \cup F_2$.
If the player queried both vertices in $A_{out}$, the oracle returns $M_3 = \{ a_4b_2, a_5b_3 \}$. Otherwise $M_3$ would consist of one or zero edges depending on the player's query. In particular, we have $M_3 \subset E_3$.

\vspace{0.01cm} 

In cases 2 and 3, we do not define any edges involving vertices from $A_{out}$ or $B_{out}$, so the 
oracle proceeds regardless of which vertices in $A_{out}, B_{out}$ the player queried.

 \item \textbf{Case 2:} Player queries $a_1, b_2, b_3$. The oracle defines $H_3 = (A, B, E_3, F_3)$ such that $E_3 = M \cup E_2 \cup \{ a_1b_2 \}$
 and $F_3 = A_{out} \times B_{out} \cup F_2 \cup \{ a_4b_3, a_5b_3 \}$. The oracle returns $M_3 = \{ a_1b_2 \}$.

  \item \textbf{Case 3:} Player queries $b_1, b_2, a_3$. The oracle defines $H_3 = (A, B, E_3, F_3)$ such that $E_3 = M \cup E_2 \cup \{ a_3b_2 \}$
and $F_3 = A_{out} \times B_{out} \cup F_2 \cup \{ a_4b_1, a_5b_1 \}$. The oracle returns $M_3 = \{ a_3b_2 \}$.
\end{enumerate}


Observe that the case $b_1, a_2, b_3 \in V_3$ is not relevant, since $a_2 b_3 \in M_2$ and Assumption~\ref{ass:2}. 
Figure~\ref{fig: h3} shows that in these three cases, $H_3$ is a structure graph and the largest matching that the player thus 
able to return is of size $3$.

If the player queries more vertices in $A_{in}$ than in $B_{in}$, we will argue that the player will not learn any 
edges connecting to vertices in $A_{out}$, and since the player then only holds edges incident to $3$ of the $5$ $A$-vertices,
the player cannot report a matching larger than of size $3$. 

If the player queries all three vertices in $A_{in}$ then he clearly cannot learn any edges connecting to $A_{out}$. If the 
player queries a vertex in $B_{in}$, note that we can match it with a vertex queried in $A_{in}$, and there will be no vertices 
left to match with vertices in $A_{out}$ (see Figure~\ref{fig: 3rnd more A_in}). Since no more non-edges are defined, 
it is easy to see that edges can be added to create a perfect matching.
\end{proof}

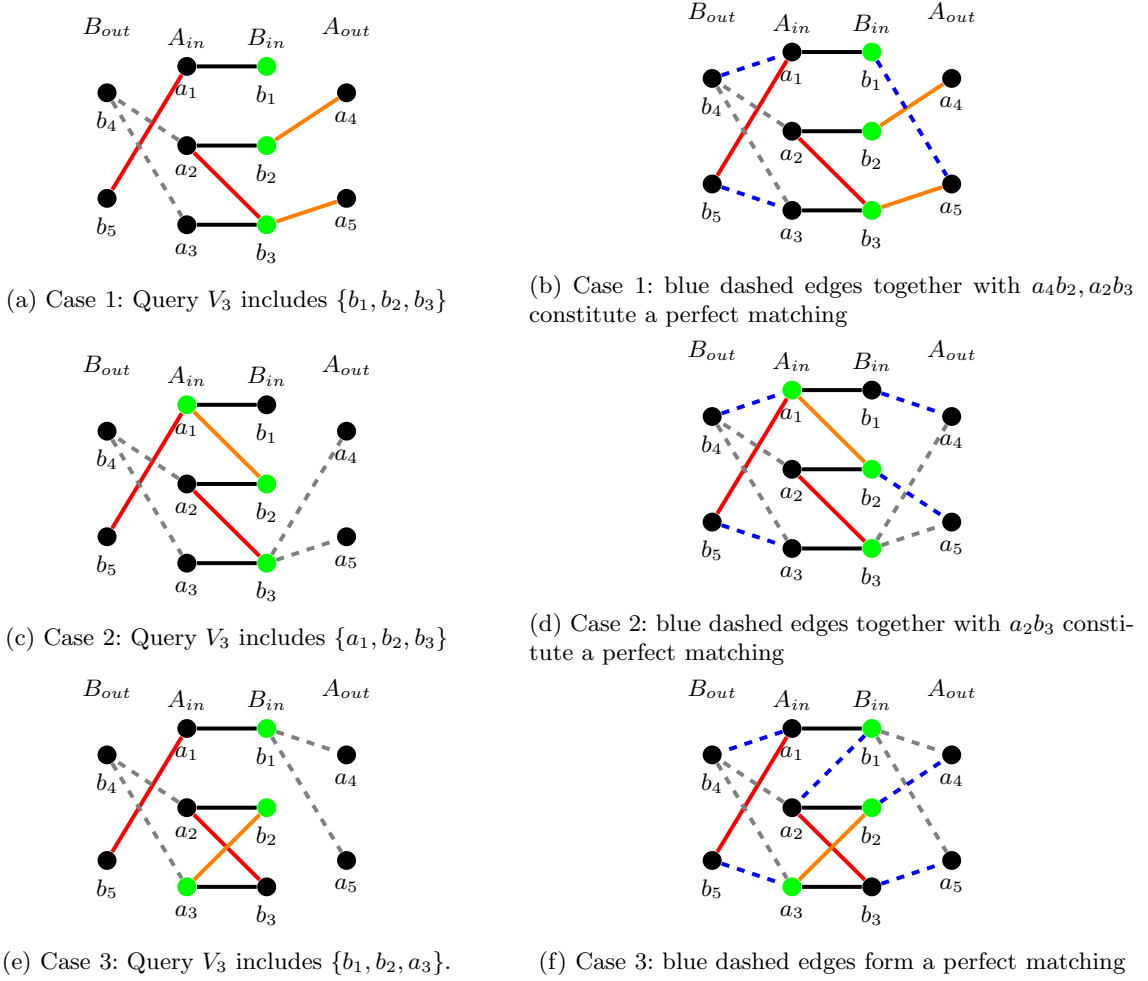
\begin{figure}[t]
\begin{subfigure}{0.5\textwidth}
\centering
\begin{tikzpicture}[scale=0.7]
\tikzstyle{vertex} = [circle, fill=black, inner sep=0pt, minimum size=0.25cm]
\tikzstyle{green vertex} = [circle, fill=green, inner sep=0pt, minimum size=0.25cm]
\tikzstyle{blue dash} = [color=blue, line width=1.5pt, dashed]
\tikzstyle{gray dash} = [color=gray, line width=1.5pt, dashed]
\tikzstyle{red line} = [line width=1.5pt, color=red]
\tikzstyle{orange line} = [line width=1.5pt, color=orange]
\tikzstyle{black line} = [line width=1.5pt, color=black]

\node [vertex][label={[label distance=0.5cm]above:$B_{out}$}, label=below:$b_4$] (v8) at (-3.5,3) {};
\node [vertex][label=above:$A_{in}$, label=below:$a_1$] (v1) at (-2,3.5) {};
\node [green vertex][label=above:$B_{in}$, label=below:$b_1$] (v2) at (-0.5,3.5) {};
\node [vertex][label=below:$a_2$] (v3) at (-2,2) {};
\node [green vertex][label=below:$b_2$] (v4) at (-0.5,2) {};
\node [vertex][label={[label distance=0.5cm]above:$A_{out}$}, label=below:$a_4$] (v9) at (1,3) {};
\node [vertex][label=below:$b_5$] (v7) at (-3.5,1) {};
\node [vertex][label=below:$a_3$] (v5) at (-2,0.5) {};
\node [green vertex][label=below:$b_3$] (v6) at (-0.5,0.5) {};
\node [vertex][label=below:$a_5$] (v10) at (1,1) {};

\draw [black line] (v1) edge (v2);
\draw [black line] (v3) edge (v4);
\draw [black line] (v5) edge (v6);

\draw [red line] (v1) edge (v7);
\draw [red line] (v3) edge (v6);

\draw [gray dash] (v8) edge (v5);
\draw [gray dash] (v3) edge (v8);

\draw [orange line] (v4) edge (v9);
\draw [orange line] (v6) edge (v10);
\end{tikzpicture}
\caption{Case 1: Query $V_3$ includes $\{b_1, b_2, b_3 \}$}
\end{subfigure}%
\begin{subfigure}{0.5\textwidth}
\centering
\begin{tikzpicture}[scale=0.7]
\tikzstyle{vertex} = [circle, fill=black, inner sep=0pt, minimum size=0.25cm]
\tikzstyle{green vertex} = [circle, fill=green, inner sep=0pt, minimum size=0.25cm]
\tikzstyle{blue dash} = [color=blue, line width=1.5pt, dashed]
\tikzstyle{gray dash} = [color=gray, line width=1.5pt, dashed]
\tikzstyle{red line} = [line width=1.5pt, color=red]
\tikzstyle{orange line} = [line width=1.5pt, color=orange]
\tikzstyle{black line} = [line width=1.5pt, color=black]

\node [vertex][label={[label distance=0.5cm]above:$B_{out}$}, label=below:$b_4$] (v8) at (-3.5,3) {};
\node [vertex][label=above:$A_{in}$, label=below:$a_1$] (v1) at (-2,3.5) {};
\node [green vertex][label=above:$B_{in}$, label=below:$b_1$] (v2) at (-0.5,3.5) {};
\node [vertex][label=below:$a_2$] (v3) at (-2,2) {};
\node [green vertex][label=below:$b_2$] (v4) at (-0.5,2) {};
\node [vertex][label={[label distance=0.5cm]above:$A_{out}$}, label=below:$a_4$] (v9) at (1,3) {};
\node [vertex][label=below:$b_5$] (v7) at (-3.5,1) {};
\node [vertex][label=below:$a_3$] (v5) at (-2,0.5) {};
\node [green vertex][label=below:$b_3$] (v6) at (-0.5,0.5) {};
\node [vertex][label=below:$a_5$] (v10) at (1,1) {};
\draw [black line] (v1) edge (v2);
\draw [black line] (v3) edge (v4);
\draw [black line] (v5) edge (v6);

\draw [red line] (v1) edge (v7);
\draw [red line] (v3) edge (v6);

\draw [gray dash] (v8) edge (v5);
\draw [gray dash] (v3) edge (v8);

\draw [orange line] (v4) edge (v9);
\draw [orange line] (v6) edge (v10);

\draw [blue dash] (v2) edge (v10);
\draw [blue dash] (v8) edge (v1);
\draw [blue dash] (v7) edge (v5);
\end{tikzpicture}
\caption{Case 1: blue dashed edges together with $a_4b_2, a_2b_3$ constitute a perfect matching}
\end{subfigure}\newline
\begin{subfigure}{0.5\textwidth}
\centering
\begin{tikzpicture}[scale=0.7]
\tikzstyle{vertex} = [circle, fill=black, inner sep=0pt, minimum size=0.25cm]
\tikzstyle{green vertex} = [circle, fill=green, inner sep=0pt, minimum size=0.25cm]
\tikzstyle{blue dash} = [color=blue, line width=1.5pt, dashed]
\tikzstyle{gray dash} = [color=gray, line width=1.5pt, dashed]
\tikzstyle{red line} = [line width=1.5pt, color=red]
\tikzstyle{orange line} = [line width=1.5pt, color=orange]
\tikzstyle{black line} = [line width=1.5pt, color=black]

\node [vertex][label={[label distance=0.5cm]above:$B_{out}$}, label=below:$b_4$] (v8) at (-3.5,3) {};
\node [green vertex][label=above:$A_{in}$, label=below:$a_1$] (v1) at (-2,3.5) {};
\node [vertex][label=above:$B_{in}$, label=below:$b_1$] (v2) at (-0.5,3.5) {};
\node [vertex][label=below:$a_2$] (v3) at (-2,2) {};
\node [green vertex][label=below:$b_2$] (v4) at (-0.5,2) {};
\node [vertex][label={[label distance=0.5cm]above:$A_{out}$}, label=below:$a_4$] (v9) at (1,3) {};
\node [vertex][label=below:$b_5$] (v7) at (-3.5,1) {};
\node [vertex][label=below:$a_3$] (v5) at (-2,0.5) {};
\node [green vertex][label=below:$b_3$] (v6) at (-0.5,0.5) {};
\node [vertex][label=below:$a_5$] (v10) at (1,1) {};
\draw [black line] (v1) edge (v2);
\draw [black line] (v3) edge (v4);
\draw [black line] (v5) edge (v6);

\draw [red line] (v1) edge (v7);
\draw [red line] (v3) edge (v6);

\draw [gray dash] (v8) edge (v5);
\draw [gray dash] (v3) edge (v8);

\draw [orange line] (v1) edge (v4);
\draw [gray dash] (v6) edge (v9);
\draw [gray dash] (v6) edge (v10);
\end{tikzpicture}
\caption{Case 2: Query $V_3$ includes $\{a_1, b_2, b_3 \}$}
\end{subfigure}%
\begin{subfigure}{0.5\textwidth}
\centering
\begin{tikzpicture}[scale=0.7]
\tikzstyle{vertex} = [circle, fill=black, inner sep=0pt, minimum size=0.25cm]
\tikzstyle{green vertex} = [circle, fill=green, inner sep=0pt, minimum size=0.25cm]
\tikzstyle{blue dash} = [color=blue, line width=1.5pt, dashed]
\tikzstyle{gray dash} = [color=gray, line width=1.5pt, dashed]
\tikzstyle{red line} = [line width=1.5pt, color=red]
\tikzstyle{orange line} = [line width=1.5pt, color=orange]
\tikzstyle{black line} = [line width=1.5pt, color=black]

\node [vertex][label={[label distance=0.5cm]above:$B_{out}$}, label=below:$b_4$] (v8) at (-3.5,3) {};
\node [green vertex][label=above:$A_{in}$, label=below:$a_1$] (v1) at (-2,3.5) {};
\node [vertex][label=above:$B_{in}$, label=below:$b_1$] (v2) at (-0.5,3.5) {};
\node [vertex][label=below:$a_2$] (v3) at (-2,2) {};
\node [green vertex][label=below:$b_2$] (v4) at (-0.5,2) {};
\node [vertex][label={[label distance=0.5cm]above:$A_{out}$}, label=below:$a_4$] (v9) at (1,3) {};
\node [vertex][label=below:$b_5$] (v7) at (-3.5,1) {};
\node [vertex][label=below:$a_3$] (v5) at (-2,0.5) {};
\node [green vertex][label=below:$b_3$] (v6) at (-0.5,0.5) {};
\node [vertex][label=below:$a_5$] (v10) at (1,1) {};
\draw [black line] (v1) edge (v2);
\draw [black line] (v3) edge (v4);
\draw [black line] (v5) edge (v6);

\draw [red line] (v1) edge (v7);
\draw [red line] (v3) edge (v6);

\draw [gray dash] (v8) edge (v5);
\draw [gray dash] (v3) edge (v8);

\draw [orange line] (v1) edge (v4);
\draw [gray dash] (v6) edge (v9);
\draw [gray dash] (v6) edge (v10);

\draw [blue dash] (v8) edge (v1);
\draw [blue dash] (v2) edge (v9);
\draw [blue dash] (v7) edge (v5);
\draw [blue dash] (v4) edge (v10);
\end{tikzpicture}
\caption{Case 2: blue dashed edges together with $a_2b_3$ constitute a perfect matching}
\end{subfigure}\newline
\begin{subfigure}{0.5\textwidth}
\centering
\begin{tikzpicture}[scale=0.7]
\tikzstyle{vertex} = [circle, fill=black, inner sep=0pt, minimum size=0.25cm]
\tikzstyle{green vertex} = [circle, fill=green, inner sep=0pt, minimum size=0.25cm]
\tikzstyle{blue dash} = [color=blue, line width=1.5pt, dashed]
\tikzstyle{gray dash} = [color=gray, line width=1.5pt, dashed]
\tikzstyle{red line} = [line width=1.5pt, color=red]
\tikzstyle{orange line} = [line width=1.5pt, color=orange]
\tikzstyle{black line} = [line width=1.5pt, color=black]

\node [vertex][label={[label distance=0.5cm]above:$B_{out}$}, label=below:$b_4$] (v8) at (-3.5,3) {};
\node [vertex][label=above:$A_{in}$, label=below:$a_1$] (v1) at (-2,3.5) {};
\node [green vertex][label=above:$B_{in}$, label=below:$b_1$] (v2) at (-0.5,3.5) {};
\node [vertex][label=below:$a_2$] (v3) at (-2,2) {};
\node [green vertex][label=below:$b_2$] (v4) at (-0.5,2) {};
\node [vertex][label={[label distance=0.5cm]above:$A_{out}$}, label=below:$a_4$] (v9) at (1,3) {};
\node [vertex][label=below:$b_5$] (v7) at (-3.5,1) {};
\node [green vertex][label=below:$a_3$] (v5) at (-2,0.5) {};
\node [vertex][label=below:$b_3$] (v6) at (-0.5,0.5) {};
\node [vertex][label=below:$a_5$] (v10) at (1,1) {};
\draw [black line] (v1) edge (v2);
\draw [black line] (v3) edge (v4);
\draw [black line] (v5) edge (v6);

\draw [red line] (v1) edge (v7);
\draw [red line] (v3) edge (v6);

\draw [gray dash] (v8) edge (v5);
\draw [gray dash] (v3) edge (v8);

\draw [gray dash] (v2) edge (v9);
\draw [gray dash] (v2) edge (v10);

\draw [orange line] (v4) edge (v5);
\end{tikzpicture}
\caption{Case 3: Query $V_3$ includes $\{b_1, b_2, a_3 \}$.}
\end{subfigure}%
\begin{subfigure}{0.5\textwidth}
\centering
\begin{tikzpicture}[scale=0.7]
\tikzstyle{vertex} = [circle, fill=black, inner sep=0pt, minimum size=0.25cm]
\tikzstyle{green vertex} = [circle, fill=green, inner sep=0pt, minimum size=0.25cm]
\tikzstyle{blue dash} = [color=blue, line width=1.5pt, dashed]
\tikzstyle{gray dash} = [color=gray, line width=1.5pt, dashed]
\tikzstyle{red line} = [line width=1.5pt, color=red]
\tikzstyle{orange line} = [line width=1.5pt, color=orange]
\tikzstyle{black line} = [line width=1.5pt, color=black]

\node [vertex][label={[label distance=0.5cm]above:$B_{out}$}, label=below:$b_4$] (v8) at (-3.5,3) {};
\node [vertex][label=above:$A_{in}$, label=below:$a_1$] (v1) at (-2,3.5) {};
\node [green vertex][label=above:$B_{in}$, label=below:$b_1$] (v2) at (-0.5,3.5) {};
\node [vertex][label=below:$a_2$] (v3) at (-2,2) {};
\node [green vertex][label=below:$b_2$] (v4) at (-0.5,2) {};
\node [vertex][label={[label distance=0.5cm]above:$A_{out}$}, label=below:$a_4$] (v9) at (1,3) {};
\node [vertex][label=below:$b_5$] (v7) at (-3.5,1) {};
\node [green vertex][label=below:$a_3$] (v5) at (-2,0.5) {};
\node [vertex][label=below:$b_3$] (v6) at (-0.5,0.5) {};
\node [vertex][label=below:$a_5$] (v10) at (1,1) {};
\draw [black line] (v1) edge (v2);
\draw [black line] (v3) edge (v4);
\draw [black line] (v5) edge (v6);

\draw [red line] (v1) edge (v7);
\draw [red line] (v3) edge (v6);

\draw [gray dash] (v8) edge (v5);
\draw [gray dash] (v3) edge (v8);

\draw [gray dash] (v2) edge (v9);
\draw [gray dash] (v2) edge (v10);

\draw [orange line] (v4) edge (v5);
\draw [blue dash] (v1) edge (v8);
\draw [blue dash] (v7) edge (v5);
\draw [blue dash] (v2) edge (v3);
\draw [blue dash] (v4) edge (v9);
\draw [blue dash] (v6) edge (v10);
\end{tikzpicture}
\caption{Case 3: blue dashed edges form a perfect matching}
\end{subfigure}
\caption{Round 3 cases. Green vertices are queried by the player in round 3. Red edges are in $E_2 \setminus E_1$, orange is $E_3 \setminus E_2$, grey is $F_3$. The blue dashed edges can be added to the graph to create a perfect matching.}
\label{fig: h3}
\end{figure} 

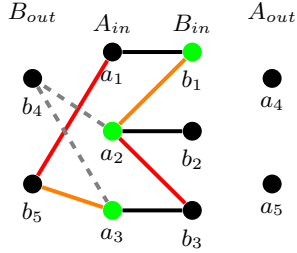
\begin{figure}[t] 
\centering
\begin{tikzpicture}[scale=0.7]
\tikzstyle{vertex} = [circle, fill=black, inner sep=0pt, minimum size=0.25cm]
\tikzstyle{green vertex} = [circle, fill=green, inner sep=0pt, minimum size=0.25cm]
\tikzstyle{blue dash} = [color=blue, line width=1.5pt, dashed]
\tikzstyle{gray dash} = [color=gray, line width=1.5pt, dashed]
\tikzstyle{red line} = [line width=1.5pt, color=red]
\tikzstyle{orange line} = [line width=1.5pt, color=orange]
\tikzstyle{black line} = [line width=1.5pt, color=black]

\node [vertex][label={[label distance=0.5cm]above:$B_{out}$}, label=below:$b_4$] (v8) at (-3.5,3) {};
\node [vertex][label=above:$A_{in}$, label=below:$a_1$] (v1) at (-2,3.5) {};
\node [green vertex][label=above:$B_{in}$, label=below:$b_1$] (v2) at (-0.5,3.5) {};
\node [green vertex][label=below:$a_2$] (v3) at (-2,2) {};
\node [vertex][label=below:$b_2$] (v4) at (-0.5,2) {};
\node [vertex][label={[label distance=0.5cm]above:$A_{out}$}, label=below:$a_4$] (v9) at (1,3) {};
\node [vertex][label=below:$b_5$] (v7) at (-3.5,1) {};
\node [green vertex][label=below:$a_3$] (v5) at (-2,0.5) {};
\node [vertex][label=below:$b_3$] (v6) at (-0.5,0.5) {};
\node [vertex][label=below:$a_5$] (v10) at (1,1) {};
\draw [black line] (v1) edge (v2);
\draw [black line] (v3) edge (v4);
\draw [black line] (v5) edge (v6);

\draw [red line] (v1) edge (v7);
\draw [red line] (v3) edge (v6);

\draw [gray dash] (v8) edge (v5);
\draw [gray dash] (v3) edge (v8);

\draw [orange line] (v3) edge (v2);
\draw [orange line] (v7) edge (v5);
\end{tikzpicture}
\caption{An example of how the oracle behaves when the player queries more vertices in $A_{in}$ than in $B_{in}$ during the third round. Green vertices are queried by the player. Red edges are in $E_2 \setminus E_1$, orange is $E_3 \setminus E_2$, gray is $F_3$. The player learns no edges incident to $A_{out}$ and can therefore only report a matching of size $3$.}
\label{fig: 3rnd more A_in}
\end{figure}

\begin{lemma}
\label{lem: 3rnd <3}
If the player queries fewer than three ``in''-vertices in the third round, then the player cannot output a matching of size larger than $3$.
\end{lemma}

\begin{proof}
We distinguish the following cases:
\begin{enumerate}
\item If the player queries no ``in'' vertices, this is obvious, and we would have $M_3 = \emptyset$.
\item If the player queries exactly one ``in'' vertex, the only possible way to obtain a larger matching than one of size $3$ 
is to find an edge incident to $b_1$, i.e., by querying $b_1$, but we can define $F_3 = A_{out} \times B_{out} \cup F_2 \cup \{ a_4b_1, a_5b_1 \}$ and then $M_3 = \emptyset$. 
\item If the player queries one vertex in $A_{in}$ and one in $B_{in}$, we can connect them by an edge, say $e$, and then $M_3 = \{e\}$ does not help increasing the size of a matching.
\item If the player queries two vertices in $A_{in}$, the player will not be able to learn any edges to vertices in $A_{out}$, and so $A_{out}$ 
remains unmatched, which implies that the player cannot return a matching of size larger than $3$.
\item If the player queries two vertices in $B_{in}$, the oracle defines $H_3$ as in Case~1 of Lemma~\ref{lem:exactly-3-in-vertices}, and the matching returned is a subset of $E_3$.
\end{enumerate}
\end{proof}



Hence we have shown that no matter what queries are made in the second and third rounds, the player cannot increase the size of the matching learned within the $10$-vertex subgraph. This then holds for a graph with $|A|=|B|=n$ where $5|n$ and the theorem follows.

\begin{theorem}
The best approximation factor achievable in three rounds is $3/5$.
\end{theorem}

\section{$(1-\epsilon)$-approximation in Bipartite Graphs Requires $\Omega(\frac{1}{\epsilon})$ Rounds} \label{sec:multi-rounds}
Let $G_c=(A, B, E)$ with $A=B=[c]$ be the {\em semi-complete} graph on $2c$ vertices, i.e., vertices $a \in A$ and $b \in B$
are connected if and only if $b \ge a$. Observe that $G_c$ has a unique perfect matching $M^* = \{(i, i) \in E \ | \ i \in [c] \}$.

Let $G$ be the disjoint union of $n/(2c)$ copies of $G_c$ (assuming for simplicity that $n$ is a multiple of $2c$). We will refer to 
a copy of $G_c$ in $G$ as a gadget. We now show that computing a $(1-\epsilon)$-approximation requires $\Omega(\frac{1}{\epsilon})$ queries on $G$.

\begin{theorem}
 Any query algorithm with approximation factor $1-\epsilon$ requires at least $\frac{1}{\epsilon} - 1$ queries, even in
 bipartite graphs.
\end{theorem}
\begin{proof}
Let $c = \frac{1}{\epsilon} - 1$. We consider the graph $G$. First, suppose that the algorithm does not compute a perfect matching 
in any of the $n/(2c)$ gadgets. Then, the computed matching is of size at most $\frac{c-1}{c} \frac{n}{2}$ and thus constitutes at best 
a $\frac{c-1}{c} = 1 - \frac{\epsilon}{1-\epsilon} < 1 - \epsilon$ approximation.
The algorithm therefore needs to compute a perfect matching in at least one gadget. Since all gadgets are disjoint, we now argue
that it requires at least $c$ queries in order to compute a perfect matching in one gadget. 
Consider thus the gadget $G_c$ and denote by $M^*$ the perfect matching in $G_c$. 
We claim that each query may produce at most one edge of the perfect matching $M^*$ in $G_c$:

Indeed, let $A' = \{a_1, a_2, \dots, a_k \} \subseteq A$ and $B' = \{b_1, b_2, \dots, b_{\ell} \} \subseteq B$ be so that $A' \cup B'$ 
is any query submitted to the oracle. Further, suppose that $a_1 < a_2 < \dots < a_k$ and $b_1 < b_2 < \dots < b_{\ell}$. 
The oracle will return the following matching $M$: 
$$M = \{ a_i b_{\ell + 1 - i} \ | \ i \in [\min \{k, \ell \} ] \} \cap E \ .$$
We will now argue that $M$ is maximal and $|M \cap M^*| \le 1$. To this end, let $j$ be the largest index such that
$a_j b_{\ell + 1 - j} \in E$, which is equivalent to $j$ being the largest index so that $a_j \le b_{\ell + 1 - j}$. 
Observe that since the $(a_i)_i$ and $(b_i)_i$ are increasing, we have $a_{j'} b_{\ell + 1 - j'} \in E \Leftrightarrow j' \le j$, which also implies that vertices $a_{j'}$ are matched,
for every $j' \le j$. Consider now a vertex $a_{q}$, for some $q > j$. Since $a_{j+1} > b_{\ell -j}$ and $a_q \ge a_{j+1}$, it follows
that there is no edge between $a_q$ and any of the unmatched $B'$-vertices $\{b_1, b_2, \dots, b_{\ell - j} \}$. This implies that the 
matching $M$ is maximal. Next, suppose that $M$ contains at least one edge from $M^*$ and let $q$ be the smallest 
index such that $a_q = b_{\ell + 1 - q}$, i.e., $(a_q, b_{\ell + 1 - q}) \in M^*$. 
Then, for any $q' > q$, we have $$a_{q'} > a_q = b_{\ell + 1 - q} > b_{\ell + 1 - q'} \ ,$$
which implies that $a_{q'} \neq b_{\ell + 1 - q'}$. Hence, at most one edge from $M^*$ is returned per query.

Last, we argue that the oracle can be made streaming-consistent: Consider any ordering of the edges so that edge $ij$ arrives 
before edge $ik$, for every $k < j$.
\end{proof}

Using the oracle described in the previous proof on a single gadget $G_{n/2}$, we obtain the following corollary:

\begin{corollary}
 Any query algorithm that produces a maximum matching requires at least $n/2$ queries (on a graph on $n$ vertices), even on bipartite graphs.
\end{corollary}

\section{Improving on $1/2$ in General Graphs Requires $\Omega(n)$ Queries} \label{sec:general-graphs}
Let $G$ be a {\em bomb graph} on $n$ ($n$ even) vertices $U \cup V$ with $|U| = |V| = [n/2]$, where $G[V]$ 
is a clique, $G[U]$ is an independent set, and $u \in U$ and $v \in V$
are connected if and only if $u = v$ ($U$ and $V$ are connected via a perfect matching). Denote by $M^*$ the perfect matching between $U$ and $V$
and by $C$ the edges of the clique $G[V]$.

In the next lemma, we show that any large matching in $G$ must contain a large number of edges from $M^*$.

\begin{lemma}\label{lem:bomb-graph}
 Let $M$ be a matching in $G$. Then: $|M|  \le \frac{n}{4} + \frac{1}{2} |M \cap M^*| \ .$
\end{lemma}
\begin{proof}
 Observe that $|M| = |M \cap M^*| + |M \cap C|$, and since there are $n/2 - |M \cap M^*|$ vertices in $V$ that are not matched
 to a vertex in $U$, we have $|M \cap C| \le (n/2 - |M \cap M^*|)/2$. Hence:
 $$|M| = |M \cap M^*| + |M \cap C| \le |M \cap M^*| + (n/2 - |M \cap M^*|)/2 = \frac{n}{4} + \frac{1}{2} |M \cap M^*| \ . $$
\end{proof}

\begin{theorem}
 Any $r$-round query algorithm on general graphs has approximation ratio at most $\frac{1}{2} + \frac{r}{n}$ (on an $n$-vertex input graph).
\end{theorem}
\begin{proof}
 Consider an arbitrary query $U' \cup V'$ so that $U' \subseteq U$ and $V' \subseteq V$. The oracle returns the following matching:
 First, the oracle arbitrarily pairs up all vertices of $V'$ except possibly one in case $|V'|$ is odd. Let $M$ denote this matching. 
 If $|V'|$ is even then $M$ is returned. Suppose now that $|V'|$ is odd and let $v \in V'$ be the vertex that is not matched in $M$.
 Then, if $v$'s partner $u \in U$ in $M^*$ is contained in $U'$, then return $M \cup \{uv \}$, otherwise return $M$. 
 
 It is easy to see that, by construction, the returned matching is maximal and contains at most one edge from $M^*$. 
 Hence, in $r$-rounds the algorithm can learn at most $r$ edges from $M^*$. By Lemma~\ref{lem:bomb-graph}, the returned matching is
 therefore of size at most $\frac{n}{4} + \frac{1}{2}r$, which constitutes a $\frac{1}{2} + \frac{r}{n}$-approximation. 
 
 The oracle can be made streaming-consistent: Consider any edge order where we first have edges $C$ in arbitrary order 
 followed by $M^*$ in arbitrary order.
\end{proof}

\section{Conclusion} \label{sec:conclusion}
In this paper, we introduced a new query model that allows us to prove lower bounds for streaming algorithms for \textsf{Maximum Matching}
that repeatedly run the \textsc{Greedy} matching algorithm on a vertex-induced subgraph of the input graph. We showed that the
three rounds algorithm \textsc{3RoundMatch} with approximation factor $0.6$ is optimal for this class of algorithms. 
We also showed that computing a $(1-\epsilon)$-approximation in bipartite graphs requires $\Omega(\frac{1}{\epsilon})$ rounds,
and computing an approximation strictly better than $\frac{1}{2}$ in general graphs requires $\Omega(n)$ rounds. We conclude with open questions:
\begin{itemize}
 \item Can we prove that computing a maximum matching in the vertex-query model in bipartite graphs requires $\Omega(n^2)$ rounds, or 
 is there an algorithm that requires only $o(n^2)$ rounds?
 \item Can we prove a $\Omega(\frac{1}{\epsilon^2})$ lower bound for computing a $(1-\epsilon)$-approximation in bipartite graphs?
\end{itemize}



\bibliography{kk20}

\appendix

\end{document}